\documentclass{eptcs}

\usepackage{multicol}
\usepackage{enumerate}
\usepackage{enumitem}

\usepackage[nottoc]{tocbibind}
\usepackage{tikzit}
\input{tikz/zxw_styles.tikzdefs}

\tikzstyle{gbox}=[box, draw=black, shape=rectangle, fill={zx_green}, tikzit fill={rgb,255: red,181; green,215; blue,181}]
\tikzstyle{hbox}=[box, draw=black, shape=rectangle, fill=yellow, minimum size=.50em]
\tikzstyle{box}=[draw=black, shape=rectangle, fill=white, minimum size=.95em, inner sep=0.15em, scale=0.85, font={\footnotesize}]
\tikzstyle{gn}=[draw=black, shape=circle, fill={zx_green}, draw=black, inner sep=0.7mm, minimum width=0pt, minimum height=0pt, tikzit fill={rgb,255: red,181; green,215; blue,181}]
\tikzstyle{rn}=[gn, fill={zx_red}, draw=black, tikzit fill={rgb,255: red,232; green,165; blue,165}]
\tikzstyle{gn_phase}=[shape=rectangle, fill={zx_green}, draw=black, minimum size=1em, rounded corners=0.4em, inner sep=0.2em, outer sep=-0.2em, scale=0.75, font={\scriptsize\boldmath}, tikzit shape=circle, tikzit fill={rgb,255: red,181; green,215; blue,181}]
\tikzstyle{rn_phase}=[{gn_phase}, fill={zx_red}, draw=black, tikzit fill={rgb,255: red,232; green,165; blue,165}]
\tikzstyle{wn}=[{gn_phase}, fill=white, draw=black]
\tikzstyle{rtriang}=[shape=isosceles triangle, fill=yellow, draw=black, isosceles triangle stretches=true, inner sep=0.8pt, minimum width=0.25cm, minimum height=2mm]
\tikzstyle{ltriang}=[rtriang, shape=isosceles triangle, fill=yellow, draw=black, shape border rotate=180]
\tikzstyle{utriang}=[rtriang, shape=isosceles triangle, fill=yellow, draw=black, shape border rotate=90]
\tikzstyle{dtriang}=[rtriang, shape=isosceles triangle, fill=yellow, draw=black, shape border rotate=-90]
\tikzstyle{lmat}=[shape=signal, signal to=west, signal from=east, fill={zx_grey}, draw=black, minimum height=6pt, inner sep=1pt, font={\scriptsize  \boldmath}, tikzit fill=gray, tikzit category=GLA]
\tikzstyle{rmat}=[shape=signal, signal to=east, signal from=west, fill={zx_grey}, draw=black, minimum height=6pt, inner sep=1pt, font={\scriptsize  \boldmath}, tikzit fill=gray, tikzit category=GLA]
\tikzstyle{dmat}=[shape=signal, signal to=west, signal from=east, fill={zx_grey}, draw=black, minimum height=6pt, inner sep=1pt, font={\scriptsize  \boldmath}, tikzit fill=gray, tikzit category=GLA, rotate=90]
\tikzstyle{umat}=[shape=signal, signal to=east, signal from=west, fill={zx_grey}, draw=black, minimum height=6pt, inner sep=1pt, font={\scriptsize  \boldmath}, tikzit fill=gray, tikzit category=GLA, rotate=90]
\tikzstyle{uw}=[shape=isosceles triangle, isosceles triangle stretches=true, fill=black, draw=black, minimum width=0.4cm, minimum height=3mm, inner sep=1pt, shape border rotate=90]
\tikzstyle{dw}=[shape=isosceles triangle, isosceles triangle stretches=true, fill=black, draw=black, minimum width=0.4cm, minimum height=3mm, inner sep=1pt, shape border rotate=-90]
\tikzstyle{lw}=[shape=isosceles triangle, isosceles triangle stretches=true, fill=black, draw=black, minimum width=0.4cm, minimum height=3mm, inner sep=1pt, shape border rotate=135]
\tikzstyle{rw}=[shape=isosceles triangle, isosceles triangle stretches=true, fill=black, draw=black, minimum width=0.4cm, minimum height=3mm, inner sep=1pt, shape border rotate=-45]
\tikzstyle{d_split}=[shape=trapezium, fill=white, draw=black, minimum width=11pt, inner sep=1pt]
\tikzstyle{d_merge}=[shape=trapezium, fill=white, draw=black, minimum width=11pt, inner sep=1pt, rotate=180]
\tikzstyle{braceedge}=[decorate, decoration={brace, amplitude=2mm, raise=-1mm}]
\tikzstyle{wire label}=[font={\scriptsize}, auto]
\tikzstyle{label}=[font={\scriptsize}, text height=1ex, text depth=0.15ex]
\tikzstyle{left label}=[label, anchor=east, xshift=1.5mm]
\tikzstyle{right label}=[label, anchor=west, xshift=-1.5mm]
\tikzstyle{dbl rn}=[gn, fill={zx_red}, draw=black, tikzit fill={rgb,255: red,232; green,165; blue,165}, line width=0.5mm]
\tikzstyle{dbl rn_phase}=[{gn_phase}, fill={zx_red}, draw=black, tikzit fill={rgb,255: red,232; green,165; blue,165}, line width=0.5mm]

\tikzstyle{blue line}=[-, fill=none, draw={rgb,255: red,69; green,69; blue,255}, tikzit draw=blue]
\tikzstyle{thick line}=[-, line width=0.5mm]
\tikzstyle{filllayer}=[-, fill={rgb,255: red,178; green,227; blue,247}, opacity=0.5]
\tikzstyle{bg}=[-, filllayer, fill={rgb,255: red,255; green,250; blue,200}, draw=lightgray]
\tikzstyle{white}=[-, filllayer, fill=white, opacity=1]

\usepackage{todonotes}

\usepackage[title, titletoc]{appendix}

\usepackage[english]{babel}

\usepackage{amsthm, amsfonts, amssymb} 
\usepackage{amsmath}
\numberwithin{equation}{section}

\usepackage{hyperref}
\usepackage{cleveref}
\usepackage{physics}
\usepackage{stmaryrd}

\usepackage{tikzit}

\newtheorem{lemma}{Lemma}[section]
\newtheorem{thm}{Theorem}[section]
\theoremstyle{definition}
\newtheorem{definition}{Definition}[section]
\newtheorem{prop}{Proposition}
\newtheorem*{prop*}{Proposition}

\newtheorem{remark}{Remark}

\newcommand{\redzero}{\begin{tikzpicture}
	\begin{pgfonlayer}{nodelayer}
		\node [style=rn] (0) at (0, 0) {};
		\node [style=none] (1) at (0, -0.5) {};
	\end{pgfonlayer}
	\begin{pgfonlayer}{edgelayer}
		\draw (0) to (1.center);
	\end{pgfonlayer}
\end{tikzpicture}
}

\newcommand{\redpi}{\begin{tikzpicture}
	\begin{pgfonlayer}{nodelayer}
		\node [style={rn_phase}] (0) at (0, 0) {$\pi$};
		\node [style=none] (1) at (0, -0.5) {};
	\end{pgfonlayer}
	\begin{pgfonlayer}{edgelayer}
		\draw (0) to (1.center);
	\end{pgfonlayer}
\end{tikzpicture}
}

\newcommand{\redzeroup}{\begin{tikzpicture}
	\begin{pgfonlayer}{nodelayer}
		\node [style={rn}] (0) at (0, 0) {};
		\node [style=none] (1) at (0, 0.5) {};
	\end{pgfonlayer}
	\begin{pgfonlayer}{edgelayer}
		\draw (0) to (1.center);
	\end{pgfonlayer}
\end{tikzpicture}
}

\newcommand{\greenzeroup}{\begin{tikzpicture}
	\begin{pgfonlayer}{nodelayer}
		\node [style={gn}] (0) at (0, 0) {};
		\node [style=none] (1) at (0, 0.5) {};
	\end{pgfonlayer}
	\begin{pgfonlayer}{edgelayer}
		\draw (0) to (1.center);
	\end{pgfonlayer}
\end{tikzpicture}
}

\newcommand{\zspid}[1][0.4]{\begin{tikzpicture}[scale=#1]
	\begin{pgfonlayer}{nodelayer}
		\node [style=gn] (0) at (0, 0) {};
		\node [style=none] (1) at (0.5, -0.5) {};
		\node [style=none] (2) at (-0.5, -0.5) {};
		\node [style=none] (3) at (0, 0.75) {};
	\end{pgfonlayer}
	\begin{pgfonlayer}{edgelayer}
		\draw (0) to (1.center);
		\draw (2.center) to (0);
		\draw (3.center) to (0);
	\end{pgfonlayer}
\end{tikzpicture}
}

\newcommand{\wspid}[1][0.4]{\begin{tikzpicture}[scale=#1]
	\begin{pgfonlayer}{nodelayer}
		\node [style=uw] (0) at (0, 0) {};
		\node [style=none] (1) at (0.5, -0.5) {};
		\node [style=none] (2) at (-0.5, -0.5) {};
		\node [style=none] (3) at (0, 0.75) {};
	\end{pgfonlayer}
	\begin{pgfonlayer}{edgelayer}
		\draw (0) to (1.center);
		\draw (2.center) to (0);
		\draw (3.center) to (0);
	\end{pgfonlayer}
\end{tikzpicture}
}

\newcommand{\coW}[1][0.4]{\begin{tikzpicture}[scale=#1]
	\begin{pgfonlayer}{nodelayer}
		\node [style=dw] (0) at (0, 0) {};
		\node [style=none] (1) at (0.5, 0.5) {};
		\node [style=none] (2) at (-0.5, 0.5) {};
		\node [style=none] (3) at (0, -0.75) {};
	\end{pgfonlayer}
	\begin{pgfonlayer}{edgelayer}
		\draw (0) to (1.center);
		\draw (2.center) to (0);
		\draw (3.center) to (0);
	\end{pgfonlayer}
\end{tikzpicture}
}

\newcommand{\zspids}[1][0.4]{\begin{tikzpicture}[scale=#1]
	\begin{pgfonlayer}{nodelayer}
		\node [style=gn] (0) at (0, 0) {};
		\node [style=none] (1) at (0.5, -0.5) {};
		\node [style=none] (2) at (-0.5, -0.5) {};
		\node [style=none] (3) at (0, 0.75) {};
            \node [style=label] (4) at (0, -0.5) {$\ldots$};
	\end{pgfonlayer}
	\begin{pgfonlayer}{edgelayer}
		\draw (0) to (1.center);
		\draw (2.center) to (0);
		\draw (3.center) to (0);
	\end{pgfonlayer}
\end{tikzpicture}
}

\newcommand{\wspids}[1][0.4]{\begin{tikzpicture}[scale=#1]
	\begin{pgfonlayer}{nodelayer}
		\node [style=uw] (0) at (0, 0) {};
		\node [style=none] (1) at (0.5, -0.5) {};
		\node [style=none] (2) at (-0.5, -0.5) {};
		\node [style=none] (3) at (0, 0.75) {};
            \node [style=label] (4) at (0, -0.5) {$\ldots$};
	\end{pgfonlayer}
	\begin{pgfonlayer}{edgelayer}
		\draw (0) to (1.center);
		\draw (2.center) to (0);
		\draw (3.center) to (0);
	\end{pgfonlayer}
\end{tikzpicture}
}

\newcommand{\coWs}[1][0.4]{\begin{tikzpicture}[scale=#1]
	\begin{pgfonlayer}{nodelayer}
		\node [style=dw] (0) at (0, 0) {};
		\node [style=none] (1) at (0.5, 0.5) {};
		\node [style=none] (2) at (-0.5, 0.5) {};
		\node [style=none] (3) at (0, -0.75) {};
		\node [style=label] (4) at (0, 0.6) {$\ldots$};
	\end{pgfonlayer}
	\begin{pgfonlayer}{edgelayer}
		\draw (0) to (1.center);
		\draw (2.center) to (0);
		\draw (3.center) to (0);
	\end{pgfonlayer}
\end{tikzpicture}}

\newcommand{\hgate}[1][0.5]{\begin{tikzpicture}[scale=#1]
	\begin{pgfonlayer}{nodelayer}
		\node [style=hbox] (0) at (0, 0) {};
		\node [style=none] (1) at (0, 0.75) {};
		\node [style=none] (2) at (0, -0.75) {};
	\end{pgfonlayer}
	\begin{pgfonlayer}{edgelayer}
		\draw (2.center) to (0);
		\draw (0) to (1.center);
	\end{pgfonlayer}
\end{tikzpicture}
}

\newcommand{\numbergate}[1][a]{\begin{tikzpicture}[scale=0.5]
	\begin{pgfonlayer}{nodelayer}
		\node [style=gbox] (0) at (0, 0) {$#1$};
		\node [style=none] (3) at (0, 0.75) {};
		\node [style=none] (4) at (0, -0.75) {};
	\end{pgfonlayer}
	\begin{pgfonlayer}{edgelayer}
		\draw (3.center) to (0);
		\draw (0) to (4.center);
	\end{pgfonlayer}
\end{tikzpicture}}

\newcommand{\zeroproj}[1][0.5]{\begin{tikzpicture}[scale=#1]
	\begin{pgfonlayer}{nodelayer}
		\node [style=rn] (0) at (0, 0) {};
		\node [style=none] (3) at (0, 0.75) {};
		\node [style=rn] (4) at (0.75, -0.75) {};
		\node [style=none] (5) at (0.75, -1.5) {};
		\node [style=label] (6) at (0, -1) {$\ldots$};
		\node [style=rn] (7) at (-0.75, -0.75) {};
		\node [style=none] (8) at (-0.75, -1.5) {};
	\end{pgfonlayer}
	\begin{pgfonlayer}{edgelayer}
		\draw (3.center) to (0);
		\draw (5.center) to (4);
		\draw (8.center) to (7);
	\end{pgfonlayer}
\end{tikzpicture}
}

\newcommand{\numberstate}[1][a]{\begin{tikzpicture}[scale=0.4]
	\begin{pgfonlayer}{nodelayer}
		\node [style=gbox] (0) at (0, 0) {$#1$};
		\node [style=none] (1) at (0, 1) {};
	\end{pgfonlayer}
	\begin{pgfonlayer}{edgelayer}
		\draw (1.center) to (0);
	\end{pgfonlayer}
\end{tikzpicture}
}

\newcommand{\idwire}[1][0.5]{\begin{tikzpicture}[scale=#1]
	\begin{pgfonlayer}{nodelayer}
		\node [style=none] (0) at (0, 0) {};
		\node [style=none] (1) at (0, -1) {};
	\end{pgfonlayer}
	\begin{pgfonlayer}{edgelayer}
		\draw (0.center) to (1.center);
	\end{pgfonlayer}
\end{tikzpicture}
}

\newcommand{\swap}[1][0.25]{\begin{tikzpicture}[scale=#1]
	\begin{pgfonlayer}{nodelayer}
		\node [style=none] (0) at (-1, -1) {};
		\node [style=none] (1) at (1, 1) {};
		\node [style=none] (2) at (1, -1) {};
		\node [style=none] (3) at (-1, 1) {};
	\end{pgfonlayer}
	\begin{pgfonlayer}{edgelayer}
		\draw [in=90, out=-90, looseness=0.75] (3.center) to (2.center);
		\draw [in=-90, out=90, looseness=0.75] (0.center) to (1.center);
	\end{pgfonlayer}
\end{tikzpicture}}

\newcommand{\ccap}[1][0.5]{\begin{tikzpicture}[scale=#1]
	\begin{pgfonlayer}{nodelayer}
		\node [style=none] (0) at (0, 0) {};
		\node [style=none] (1) at (1, 0) {};
	\end{pgfonlayer}
	\begin{pgfonlayer}{edgelayer}
		\draw [bend left=90, looseness=1.75] (0.center) to (1.center);
	\end{pgfonlayer}
\end{tikzpicture}}

\newcommand{\ccup}[1][0.5]{\begin{tikzpicture}[scale=#1]
	\begin{pgfonlayer}{nodelayer}
		\node [style=none] (0) at (0, 0) {};
		\node [style=none] (1) at (1, 0) {};
	\end{pgfonlayer}
	\begin{pgfonlayer}{edgelayer}
		\draw [bend right=90, looseness=1.75] (0.center) to (1.center);
	\end{pgfonlayer}
\end{tikzpicture}}

\newcommand{\ddiag}[1][0.5]{\begin{tikzpicture}[scale=#1]
	\begin{pgfonlayer}{nodelayer}
		\node [style=none] (0) at (-1, 0.5) {};
		\node [style=none] (1) at (1, 0.5) {};
		\node [style=none] (2) at (-1, -0.5) {};
		\node [style=none] (3) at (1, -0.5) {};
		\node [style=none] (4) at (0, 0.5) {};
		\node [style=none] (5) at (-0.5, -0.5) {};
		\node [style=none] (6) at (0.5, -0.5) {};
		\node [style=none] (7) at (0, 0) {$D$};
		\node [style=none] (8) at (-0.5, -1.5) {};
		\node [style=none] (9) at (0, 1.5) {};
		\node [style=none] (10) at (0.5, -1.5) {};
		\node [style=label] (11) at (0, -1.25) {$\ldots$};
	\end{pgfonlayer}
	\begin{pgfonlayer}{edgelayer}
		\draw (0.center) to (1.center);
		\draw (1.center) to (3.center);
		\draw (3.center) to (2.center);
		\draw (2.center) to (0.center);
		\draw (5.center) to (8.center);
		\draw (6.center) to (10.center);
		\draw (4.center) to (9.center);
	\end{pgfonlayer}
\end{tikzpicture}}

\newcommand{\xsq}[1][0.4]{\begin{tikzpicture}[scale=#1]
	\begin{pgfonlayer}{nodelayer}
		\node [style=none] (0) at (0, -1.5) {};
		\node [style=dw] (1) at (0, -0.75) {};
		\node [style=gn] (2) at (0, 0.5) {};
		\node [style=none] (3) at (0, 1.25) {};
	\end{pgfonlayer}
	\begin{pgfonlayer}{edgelayer}
		\draw (3.center) to (2);
		\draw [bend left] (2) to (1);
		\draw [bend right] (2) to (1);
		\draw (1) to (0.center);
	\end{pgfonlayer}
\end{tikzpicture}
}

\newcommand{\eqq}[1]{~\overset{\left(#1\right)}{=}~}
\newcommand{\linesep}{10pt}
\newcommand{\neweqline}{\displaybreak[0] \\[\linesep]}

\newcommand{\ketz}{|0\rangle}
\newcommand{\keto}{|1\rangle}
\newcommand{\braz}{\langle0|}
\newcommand{\brao}{\langle1|}

\newcommand{\polyring}[1][n]{\mathcal{P}_{#1}}
\newcommand{\csring}[1][n]{\tilde{S}_{#1}}

\usepackage{multicol, amsthm, amsfonts}

\definecolor{zx_grey}{RGB}{211,211,211}
\definecolor{zx_pink}{RGB}{255,200,240}
\definecolor{zx_green}{RGB}{216,248,216}
\definecolor{zx_green2}{RGB}{180,248,180}
\definecolor{zx_red}{RGB}{232,165,165}

\newcommand{\lowerbox}[2][5]{\raisebox{-#1pt}{#2}}

\newcommand{\Control}{\textsf{Ctrl}}

\title{Algebraic Structure of Quantum Controlled States and Operators}
\author{
Edwin Agnew 
\institute{Department of Computer Science\\ University of Oxford}
\and 
Lia Yeh
\institute{Department of Computer Science and Technology\\ University of Cambridge\footnote{This work was done while LY was at the University of Oxford.}}
\email{\quad ly404@cam.ac.uk}
\and 
Richie Yeung
\institute{Department of Computer Science\\ University of Oxford}
\email{\quad richie.yeung@cs.ox.ac.uk}
}
\def\titlerunning{Algebraic Structure of Quantum Controlled States and Operators}
\def\authorrunning{E. Agnew, L. Yeh, R. Yeung}

\begin{document}

\maketitle

\begin{abstract}
    Quantum control is an important logical primitive of quantum computing programs, and an important concept for equational reasoning in quantum graphical calculi.
    We show that controlled diagrams in the ZXW-calculus admit rich algebraic structure. The perspective of the higher-order map $\Control$ recovers the standard notion of quantum controlled gates, while respecting sequential and parallel composition and multiple-control.

    In this work, we prove that controlled square matrices form a ring and therefore satisfy powerful rewrite rules. We also show that controlled states form a ring isomorphic to multilinear polynomials. Putting these together, we have completeness for polynomials over same-size square matrices.
    These properties supply new rewrite rules that make factorisation of arbitrary qubit Hamiltonians achievable inside a single graphical calculus.
\end{abstract}

\section{Introduction}
Controlling or branching to different possible linear maps, relations, or channels is important across quantum information and quantum computation, and has been studied through many different approaches. In quantum algorithms common techniques are block encodings~\cite{gilyen2019qsvt, rall2020qalgsblockenc} and linear combination of unitaries~\cite{childs2012hamsimlcu}, while a number of formalisations have included routed quantum circuits~\cite{vanrietvelde2021routed}, the many-worlds calculus~\cite{chardonnet2023manyworlds}, categorifying signal flow diagrams~\cite{baez2015categoriesctrl}, and classical and quantum control in quantum modal logic~\cite{sati2023quantummonadology}.

The question we are interested in is how quantum graphical calculi such as the ZX~\cite{coecke2011zx}, ZW~\cite{coecke2010zw}, and ZH~\cite{backens2019zh} calculus can be augmented to support properties of quantum control.
An early use of controlled state diagrams was for proving constructive and rational angle ZX calculus completeness~\cite{jeandel2018zxconstructive}. More recently, controlled state and controlled matrix diagrams have been applied to addition and differentiation of ZX diagrams~\cite{jeandel2024adddiffzx}, differentiating and integrating ZX diagrams for quantum machine learning~\cite{wang2022diffintzx}, Hamiltonian exponentiation and simulation~\cite{shaikh2022sum}, non-linear optical quantum computing~\cite{de2023light}, and fermion-to-qubit mappings~\cite{mcdowallrose2025f2q}. To sum ZX diagrams, these works have used controlled states along with the W generator from the ZW calculus.

Given how useful controlled diagrams are, a natural question to ask is why they work: What their underlying mathematical structures are, and which equational rewrites they satisfy.
Before descending into the details, we present the relationship between controlled matrices, and the conventional notion of quantum controlled gates. We show how to recover the latter from the former, through the higher-order map $\Control$ that maps square matrices to controlled square matrices.

First of our main results, we show that the set of all controlled $n$-partite states defines a commutative ring. We introduce $\boxplus$ which defines an Abelian group and $\boxtimes$ which defines a commutative monoid, and show that $\boxtimes$ distributes over $\boxplus$. The fragment of the qubit ZW calculus corresponding to controlled states, which we call \emph{arithmetic diagrams} hence defines a ring which we prove is isomorphic to multilinear polynomials $\mathbb{C}[x_1,...,x_n]/({x_1}^2,...,{x_n}^2)$. We prove completeness for all operations in this ring by an algorithm to rewrite any arithmetic diagram to what we call \emph{polynormal form (PNF)}. Likewise, we show that the set of all controlled square matrices on $n$ qubits defines a non-commutative ring $(\tilde{M^n}, \lowerbox[10]{\wspids}, \lowerbox[10]{\zspids})$, and we prove a second completeness result for this ring.

To the qubit ZXW-calculus, whose completeness for arbitrary finite dimension was proven in Ref.~\cite{poor2023completeness}, we add controlled states and controlled square matrices as new generators, along with the rewrite rules for completeness over their rings.
Now in the same calculus, we plug controlled states into each control wire of controlled square matrices, and show using only the rewrite rules so far that this is isomorphic to multivariate polynomials over same-size square matrices. Commutativity of controlled square matrices holds in the special case that the controls target mutually exclusive sectors, allowing copying of arbitrary controlled diagrams. As a result, we can factor multivariate polynomials over same-size square matrices. This third completeness result means we now have the ability to factor any Hamiltonian in the ZXW-calculus~\cite{shaikh2022sum}, even with all its terms black-boxed.
In sum, these algebraic properties of quantum control give rise to powerful new graphical reasoning capabilities.

\section{Preliminaries}

This section introduces the ZXW-calculus, and how controlled diagrams are defined in it. The ZXW-calculus is a graphical formalism for qudit computation, unifying the ZX and ZW calculi and synthesising their relative strengths. The ZXW-calculus consists of diagrams built from a small number of generators and equipped with a complete set of rewrite rules, which enables all equalities between linear maps to be proven diagrammatically.  Diagrams are to be read top to bottom and left to right.

\subsection{The ZXW-Calculus}
The qubit ZXW-calculus is built from the following generators:

\begin{gather}
  \left\llbracket \quad \lowerbox{\idwire[0.5]} \;\;\; \right\rrbracket ~=~ \begin{bmatrix}1 & 0 \\ 0 & 1\end{bmatrix} \quad
  \left\llbracket \; \lowerbox[10]{\swap[0.4]} \; \right\rrbracket ~=~ \begin{bmatrix} 1 & 0 & 0 & 0 \\ 0 & 0 & 1 & 0 \\ 0 & 1 & 0 & 0 \\ 0 & 0 & 0 & 1\end{bmatrix} \quad
  \Big\llbracket \; \lowerbox[3]{\ccap} \; \Big\rrbracket ~=~ \begin{bmatrix} 1 \\ 0 \\ 0 \\ 1\end{bmatrix} \quad
  \left\llbracket \; \lowerbox[5]{\ccup} \; \right\rrbracket ~=~ \begin{bmatrix} 1 & 0 & 0 & 1 \end{bmatrix} \\
  \left\llbracket \;\; \begin{tikzpicture}
	\begin{pgfonlayer}{nodelayer}
		\node [style=gbox] (0) at (0, 0) {$c$};
		\node [style=none] (1) at (-1, 1) {};
		\node [style=none] (2) at (1, 1) {};
		\node [style=none] (3) at (-1, -1) {};
		\node [style=none] (4) at (1, -1) {};
		\node [style=none] (5) at (-1, -1) {};
		\node [style=none] (6) at (-1, -1) {};
		\node [style=label] (7) at (0, -0.75) {$\ldots$};
		\node [style=label] (8) at (0, 0.75) {$\ldots$};
		\node [style=label] (9) at (0, 1.25) {$n$};
		\node [style=label] (10) at (0, -1.25) {$m$};
	\end{pgfonlayer}
	\begin{pgfonlayer}{edgelayer}
		\draw [bend left=15] (0) to (1.center);
		\draw [bend right=15] (0) to (2.center);
		\draw [bend left=15] (6.center) to (0);
		\draw [bend left=15] (0) to (4.center);
	\end{pgfonlayer}
\end{tikzpicture}
 \;\; \right\rrbracket ~=~ |0^m\rangle\langle0^n| + c|1^m\rangle\langle1^n|, c \in \mathbb{C} \qquad
  \left\llbracket \;\; \raisebox{-10pt}{\wspid[0.7]} \;\; \right\rrbracket ~=~ |00\rangle \braz + |01\rangle \brao + |10\rangle \brao \\
  \left\llbracket \;\;\; \raisebox{-8pt}{\hgate} \;\; \right\rrbracket = \frac{1}{\sqrt{2}}\begin{bmatrix}1 & 1 \\ 1 & -1\end{bmatrix}
\end{gather}

For simplicity, we introduce the following additional notation:

\begin{gather}
  \begin{tikzpicture}
	\begin{pgfonlayer}{nodelayer}
		\node [style=gbox] (0) at (1.5, 0) {$e^{i\alpha}$};
		\node [style=none] (1) at (2.5, 1) {};
		\node [style=none] (2) at (0.5, 1) {};
		\node [style=none] (3) at (2.5, -1) {};
		\node [style=none] (4) at (0.5, -1) {};
		\node [style=none] (5) at (2.5, -1) {};
		\node [style=none] (6) at (2.5, -1) {};
		\node [style=label] (7) at (1.5, -0.75) {$\ldots$};
		\node [style=label] (8) at (1.5, 0.75) {$\ldots$};
		\node [style={gn_phase}] (9) at (-1.75, 0) {$\alpha$};
		\node [style=none] (10) at (-0.75, 1) {};
		\node [style=none] (11) at (-2.75, 1) {};
		\node [style=none] (12) at (-0.75, -1) {};
		\node [style=none] (13) at (-2.75, -1) {};
		\node [style=none] (14) at (-0.75, -1) {};
		\node [style=none] (15) at (-0.75, -1) {};
		\node [style=label] (16) at (-1.75, -0.75) {$\ldots$};
		\node [style=label] (17) at (-1.75, 0.75) {$\ldots$};
		\node [style=none] (18) at (-0.25, 0) {$:=$};
		\node [style={gn_phase}] (19) at (11.75, 0) {$0$};
		\node [style=none] (20) at (10.75, 1) {};
		\node [style=none] (21) at (12.75, 1) {};
		\node [style=none] (22) at (10.75, -1) {};
		\node [style=none] (23) at (12.75, -1) {};
		\node [style=none] (24) at (10.75, -1) {};
		\node [style=none] (25) at (10.75, -1) {};
		\node [style=label] (26) at (11.75, -0.75) {$\ldots$};
		\node [style=label] (27) at (11.75, 0.75) {$\ldots$};
		\node [style=gn] (28) at (8, 0) {};
		\node [style=none] (29) at (7, 1) {};
		\node [style=none] (30) at (9, 1) {};
		\node [style=none] (31) at (7, -1) {};
		\node [style=none] (32) at (9, -1) {};
		\node [style=none] (33) at (7, -1) {};
		\node [style=none] (34) at (7, -1) {};
		\node [style=label] (35) at (8, -0.75) {$\ldots$};
		\node [style=label] (36) at (8, 0.75) {$\ldots$};
		\node [style=none] (37) at (9.75, 0) {$:=$};
		\node [style=none] (38) at (13.5, 0) {$=$};
		\node [style=gbox] (39) at (15.25, 0) {$1$};
		\node [style=none] (40) at (14.25, 1) {};
		\node [style=none] (41) at (16.25, 1) {};
		\node [style=none] (42) at (14.25, -1) {};
		\node [style=none] (43) at (16.25, -1) {};
		\node [style=none] (44) at (14.25, -1) {};
		\node [style=none] (45) at (14.25, -1) {};
		\node [style=label] (46) at (15.25, -0.75) {$\ldots$};
		\node [style=label] (47) at (15.25, 0.75) {$\ldots$};
	\end{pgfonlayer}
	\begin{pgfonlayer}{edgelayer}
		\draw [bend right=15] (0) to (1.center);
		\draw [bend left=15] (0) to (2.center);
		\draw [bend right=15] (6.center) to (0);
		\draw [bend right=15] (0) to (4.center);
		\draw [bend right=15] (9) to (10.center);
		\draw [bend left=15] (9) to (11.center);
		\draw [bend right=15] (15.center) to (9);
		\draw [bend right=15] (9) to (13.center);
		\draw [bend left=15] (19) to (20.center);
		\draw [bend right=15] (19) to (21.center);
		\draw [bend left=15] (25.center) to (19);
		\draw [bend left=15] (19) to (23.center);
		\draw [bend left=15] (28) to (29.center);
		\draw [bend right=15] (28) to (30.center);
		\draw [bend left=15] (34.center) to (28);
		\draw [bend left=15] (28) to (32.center);
		\draw [bend left=15] (39) to (40.center);
		\draw [bend right=15] (39) to (41.center);
		\draw [bend left=15] (45.center) to (39);
		\draw [bend left=15] (39) to (43.center);
	\end{pgfonlayer}
\end{tikzpicture}
\\
  \begin{tikzpicture}
	\begin{pgfonlayer}{nodelayer}
		\node [style={gn_phase}] (0) at (1.75, 0) {$\alpha$};
		\node [style=hbox] (1) at (0.75, 0.75) {};
		\node [style=hbox] (2) at (2.75, 0.75) {};
		\node [style=none] (3) at (0.75, -0.75) {};
		\node [style=hbox] (4) at (2.75, -0.75) {};
		\node [style=none] (5) at (0.75, -0.75) {};
		\node [style=hbox] (6) at (0.75, -0.75) {};
		\node [style=label] (7) at (1.75, -0.75) {$\ldots$};
		\node [style=label] (8) at (1.75, 0.75) {$\ldots$};
		\node [style={rn_phase}] (9) at (-2, 0) {$\alpha$};
		\node [style=none] (10) at (-3, 1) {};
		\node [style=none] (11) at (-1, 1) {};
		\node [style=none] (12) at (-3, -1) {};
		\node [style=none] (13) at (-1, -1) {};
		\node [style=none] (14) at (-3, -1) {};
		\node [style=none] (15) at (-3, -1) {};
		\node [style=label] (16) at (-2, -0.75) {$\ldots$};
		\node [style=label] (17) at (-2, 0.75) {$\ldots$};
		\node [style=none] (18) at (-0.25, 0) {$:=$};
		\node [style=none] (19) at (0.75, 1.5) {};
		\node [style=none] (20) at (2.75, 1.5) {};
		\node [style=none] (21) at (0.75, -1.5) {};
		\node [style=none] (22) at (2.75, -1.5) {};
	\end{pgfonlayer}
	\begin{pgfonlayer}{edgelayer}
		\draw [bend left=15] (0) to (1);
		\draw [bend right=15] (0) to (2);
		\draw [bend left=15] (6) to (0);
		\draw [bend left=15] (0) to (4);
		\draw [bend left=15] (9) to (10.center);
		\draw [bend right=15] (9) to (11.center);
		\draw [bend left=15] (15.center) to (9);
		\draw [bend left=15] (9) to (13.center);
		\draw (19.center) to (1);
		\draw (20.center) to (2);
		\draw (6) to (21.center);
		\draw (4) to (22.center);
	\end{pgfonlayer}
\end{tikzpicture}
\qquad \qquad \qquad \qquad
  \raisebox{-10pt}{\coWs[0.7]} ~:=~ \begin{tikzpicture}
	\begin{pgfonlayer}{nodelayer}
		\node [style=uw] (0) at (0, 0) {};
		\node [style=none] (1) at (0, 1) {};
		\node [style=none] (2) at (-0.75, -1) {};
		\node [style=none] (3) at (0.75, -1) {};
		\node [style=none] (4) at (2.25, 1.25) {};
		\node [style=none] (5) at (-1.75, -2) {};
		\node [style=label] (6) at (0, -0.75) {$\cdots$};
		\node [style=none] (7) at (1.25, 1.25) {};
		\node [style=none] (8) at (2.25, -1) {};
		\node [style=none] (10) at (1.25, -1) {};
		\node [style=none] (11) at (-1.75, 1) {};
	\end{pgfonlayer}
	\begin{pgfonlayer}{edgelayer}
		\draw (1.center) to (0);
		\draw [bend right=15] (0) to (2.center);
		\draw [bend left=15] (0) to (3.center);
		\draw (8.center) to (4.center);
		\draw [bend right=90] (3.center) to (8.center);
		\draw (10.center) to (7.center);
		\draw [bend right=90, looseness=0.75] (2.center) to (10.center);
		\draw [bend left=90, looseness=0.75] (11.center) to (1.center);
		\draw (11.center) to (5.center);
	\end{pgfonlayer}
\end{tikzpicture}

\end{gather}

Equations in ZXW apply diagrammatic rewrite rules which prove equalities of the underlying matrices. Omitted from Figure~\ref{fig:zxw_rules} are the \emph{structural rules} ubiquitous to quantum graphical calculi, such as that swapping twice yields the two-qubit identity and that an S-shaped cup and cap pair yields the one-qubit identity. These collectively are referred to as the \emph{Only Connectivity Matters} rule, which governs that in this calculus the following operations preserve semantic equality: So long as all connections are preserved between all inputs, outputs, and generators of the diagram, the spatial position and orientation of each generator can be varied, and wires are free to cross or bend as they please.
The Z and X generators are symmetric with respect to swapping two inputs, while the Z, X, and W generators are symmetric with respect to swapping two outputs:
\begin{equation}
  \scalebox{0.9}{\input{tikz/axioms/wsymetrydit.tikz}}
  \tag{Sym}\label{rule:Sym}
\end{equation}
However, unlike the Z and X generators, the W generator is not symmetric with respect to swapping an input and an output due to the following:
\begin{equation}
  \scalebox{0.9}{\input{tikz/axioms/w-asym.tikz}}
  \tag{Asym}\label{rule:Asym}
\end{equation}
For this reason, care must be taken that exactly one wire of each W generator is unambiguously its exactly one input, drawn aligned with one point of the triangle.

The complete rule set of the qubit ZXW-calculus from Ref.~\cite{poor2023completeness}, and our new rules for controlled states and controlled square matrices, are presented in Figure \ref{fig:zxw_rules}. Several important lemmas are found in Appendix \ref{sec:applem}.

\begin{figure}[htbp]
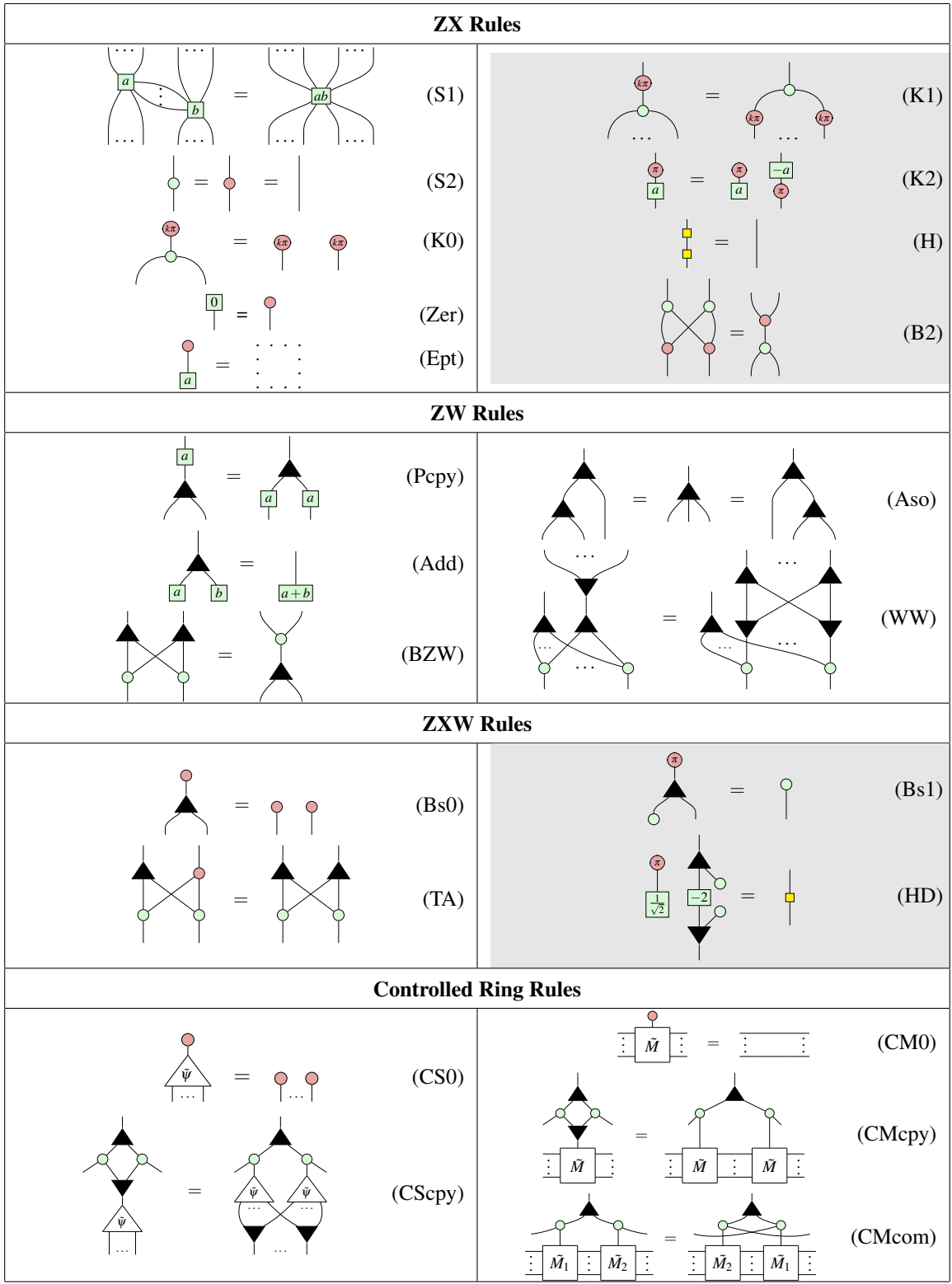

  \centering
  \renewcommand{\arraystretch}{1.5}
  \begin{tabular}{|p{0.45\textwidth}|p{0.45\textwidth}|}
    \hline
    \multicolumn{2}{|c|}{\textbf{ZX Rules}} \\
    \hline
    \begin{minipage}{\linewidth}
      \vspace{-1em}
      \begin{gather}
        \scalebox{0.9}{\input{tikz/axioms/s1.tikz}}
        \tag{S1}\label{rule:S1}\\
        \scalebox{0.9}{\input{tikz/axioms/s2.tikz}}
        \tag{S2}\label{rule:S2} \\
        \scalebox{0.9}{\input{tikz/axioms/k0copy.tikz}}
        \tag{K0}\label{rule:K0} \\
        \scalebox{0.9}{\input{tikz/axioms/zerotoreddit0.tikz}}
        \tag{Zer}\label{rule:Zer} \\
        \scalebox{0.9}{\input{tikz/axioms/rdotaemptydit0.tikz}}
        \tag{Ept}\label{rule:Ept}
      \end{gather}
    \end{minipage} &
    \noindent\colorbox{gray!20}{%
    \begin{minipage}{\linewidth}
      \vspace{-1em}
      \begin{gather}
        \scalebox{0.9}{\input{tikz/axioms/pimultiplecpdit.tikz}}
        \tag{K1}\label{rule:K1} \\
        \scalebox{0.9}{\input{tikz/axioms/k2adit.tikz}}
        \tag{K2}\label{rule:K2} \\
        \scalebox{0.9}{\input{tikz/axioms/h_id.tikz}}
        \tag{H}\label{rule:H} \\
        \scalebox{0.9}{\input{tikz/axioms/b2.tikz}}
        \tag{B2}\label{rule:B2}
      \end{gather}
    \end{minipage}} \\
    \hline
    \multicolumn{2}{|c|}{\textbf{ZW Rules}} \\
    \hline
    \begin{minipage}{\linewidth}
      \vspace{-1em}
      \begin{gather}
        \scalebox{0.9}{\input{tikz/axioms/phasecopydit.tikz}}
        \tag{Pcpy}\label{rule:Pcpy} \\
        \scalebox{0.9}{\input{tikz/axioms/additiondit.tikz}}
        \tag{Add}\label{rule:Add} \\
        \scalebox{0.9}{\input{tikz/axioms/w-bialgebra.tikz}}
        \tag{BZW}\label{rule:BZW}
      \end{gather}
    \end{minipage} &
    \begin{minipage}{\linewidth}
      \vspace{-1em}
      \begin{gather}
        \scalebox{0.9}{\input{tikz/axioms/associatedit.tikz}}
        \tag{Aso}\label{rule:Aso} \\
        \scalebox{0.9}{\input{tikz/axioms/w-w-algebra.tikz}}
        \tag{WW}\label{rule:WW}
      \end{gather}
    \end{minipage} \\
    \hline
    \multicolumn{2}{|c|}{\textbf{ZXW Rules}} \\
    \hline
    \begin{minipage}{\linewidth}
      \vspace{-1em}
      \begin{gather}
        \scalebox{0.9}{\input{tikz/axioms/triangleocopydit.tikz}}
        \tag{Bs0}\label{rule:Bs0} \\
        \scalebox{0.9}{\input{tikz/axioms/trialgebra.tikz}}
        \tag{TA}\label{rule:TA}
      \end{gather}
    \end{minipage} &
    \noindent\colorbox{gray!20}{%
    \begin{minipage}{\linewidth}
      \vspace{-1em}
      \begin{gather}
        \scalebox{0.9}{\input{tikz/axioms/trianglepicopydit2.tikz}}
        \tag{Bs1}\label{rule:Bs1}\\
        \scalebox{0.9}{\input{tikz/axioms/hadamard-decomposition2.tikz}}
        \tag{HD}\label{rule:HD}
      \end{gather}
    \end{minipage}} \\ 
    \hline
    \multicolumn{2}{|c|}{\textbf{Controlled Ring Rules}} \\
    \hline
    \begin{minipage}{\linewidth}
      \vspace{-1em}
      \begin{gather}
        \input{tikz/con/c_state0.tikz}
        \tag{CS0}\label{rule:cstate0} \\
        \scalebox{0.85}{\input{tikz/con/cs_copy.tikz}}
        \tag{CScpy}\label{rule:CScpy}
      \end{gather}
    \end{minipage} &
    \begin{minipage}{\linewidth}
      \vspace{-1em}
      \begin{gather}
        \scalebox{0.75}{\input{tikz/con/c_sq0.tikz}}
        \tag{CM0}\label{rule:c_sq0} \\
        \scalebox{0.75}{\input{tikz/con/csq_copy.tikz}}
        \tag{CMcpy}\label{rule:CMcpy}\\
        \scalebox{0.75}{\input{tikz/con/csq_add_comm_statement.tikz}}
        \tag{CMcom}\label{rule:CMcom}
      \end{gather}
    \end{minipage} \\
    \hline
  \end{tabular}
  \caption{These ZX, ZW, and ZXW Rules are altogether complete for qubit linear maps~\cite{poor2023completeness}, where $k \in \{0, 1\}$ and $a \in \mathbb{C}$.
  The white background ZX, ZW, and ZXW Rules here suffice for completeness of \emph{arithmetic diagrams} (Definition~\ref{def:arithmetic}), where \eqref{rule:TA} was used only to prove Lemma~\ref{lem:kill_quad} and we did not use the X spider case of \eqref{rule:S2}. Equivalently, these complete rules for arithmetic diagrams can all be derived from the minimal qubit ZW-calculus of~\cite{deVisme2025minzw}. Adding to these the \ref{rule:CScpy}, \ref{rule:CMcpy}, and \ref{rule:CMcom} rules, we can show that controlled states and controlled operators form rings, culminating in Theorem~\ref{thm:ctrl_pnf} achieving completeness and minimality for all operations over these rings. We did not use the gray background rules in this work.}
  \label{fig:zxw_rules}
\end{figure}

\subsection{Controlled Diagrams}

Following \cite{shaikh2022sum}, we cover the definitions of controlled states and controlled square matrices, and arithmetic on them.
Note that this is a different definition of controlled states to Ref.~\cite{jeandel2024adddiffzx} in which controlling on $\ket{0}$ is $\ket{+}^{\otimes n}$ instead of $\ket{0}^{\otimes n}$; this choice appears to make a substantial difference in the algebraic properties, which we discuss in Remark~\ref{remark:entrywise}.

\begin{definition}\label{def:ctrlsqmat}
  For an arbitrary $n \times n$ matrix $M$, we define the controlled matrix of $M$ as the diagram $\tilde{M}$ with the following interpretation:
\begin{equation}
  \left\llbracket \quad \input{tikz/con/c_sq_def.tikz} \;\;\; \right\rrbracket ~=~ \bra{0} \otimes I + \bra{1} \otimes M ~=~ \begin{bmatrix} I & M \end{bmatrix} 
\end{equation}
\end{definition}

We represent the additional dimension of $\tilde{M}$ as a vertical wire to distinguish it as the control wire. Controlled matrices satisfy the two equations:

\begin{equation*}
    \input{tikz/con/c_sq0.tikz}
    \tag{\ref{rule:c_sq0}}
\end{equation*}
\begin{equation}
    \input{tikz/con/c_sq1.tikz}
\end{equation}

We define controlled states similarly.
\begin{definition}
  For an arbitrary $n$-qubit state $\psi$, the controlled state of $\psi$ is the diagram with the following interpretation:
  
  \[
 \left\llbracket \;\;\; \input{tikz/con/c_state_def.tikz} \;\;\; \right\rrbracket ~=~ \left[ \begin{array}{cc|cc} \cline{2-3} \; 1 \; & & & \\ \; 0 \; & & & \\ \; \vdots \; & \multicolumn{2}{c}{\psi} \\ & \multicolumn{2}{c}{} \\ \; 0 \; & & \\ \cline{2-3} \end{array} \right]
 \]

\end{definition}

Controlled states satisfy the equations:
  \begin{equation*}
    \input{tikz/con/c_state0.tikz}
    \tag{\ref{rule:cstate0}}
  \end{equation*}
  \begin{equation}
    \input{tikz/con/c_state1.tikz}
    \label{eq:cstate1}
  \end{equation}

\begin{prop}[Propositions 3.3 and 3.4 of \cite{shaikh2022sum}]\label{prop:cmat_ops}
    Given controlled matrices $\tilde{M_1}, ..., \tilde{M_k}$ and $c_1, ..., c_k \in \mathbb{C}$, the controlled square matrices $\widetilde{\Pi_i M_i}$ and $\widetilde{\Sigma_i c_i M_i}$ are respectively given by
    \begin{equation}
      \input{tikz/con/csq_prod.tikz}\qquad\qquad\qquad\qquad\input{tikz/con/csq_sum.tikz}
    \end{equation}
\end{prop}

The addition and multiplication of controlled states are defined similarly to controlled matrix arithmetic, except that a layer of $\lowerbox{\coW}$s are appended at the bottom to preserve the number of outputs.
The role of $\lowerbox{\coW}$ is to \textit{copy} controlled diagrams, as we will show in Section~\ref{sec:ring}.

\begin{prop}
  Given controlled states $\tilde{\psi}$ and $\tilde{\phi}$, we define addition $\tilde{\psi} \boxplus \tilde{\phi}$ and multiplication $\tilde{\psi} \boxtimes \tilde{\phi}$ operations on them to result in the controlled states:
  \begin{equation}
    \input{tikz/con/cs_add_def.tikz} \qquad\qquad\qquad\qquad \input{tikz/con/cs_times_def.tikz}
\end{equation}
\end{prop}
\begin{remark}
  Ref.~\cite{shaikh2022sum} defined this addition with red spiders at the bottom instead of $\lowerbox{\coW}$s; the linear map is the same, being $\widetilde{\phi + \psi}$. In choosing to use $\lowerbox{\coW}$, we will soon define the arithmetic fragment of the ZW-calculus.
  
  Removing the $\lowerbox{\coW}$'s from the bottom of this multiplication gives the controlled diagram for the tensor product in Hilbert space $\widetilde{\psi \otimes \phi}$, as noted in Ref.~\cite{jeandel2024adddiffzx}.
\end{remark}
Although the interpretation of $\tilde{\psi} \boxtimes \tilde{\phi}$ is not the controlled multiplication of the linear maps $\psi$ and $\phi$, nor as nice an expression in terms of $\psi$ and $\phi$, we will show in this work that this is in fact multiplication in the ring of \emph{multilinear polynomials}.
\section{Quantum Control as a Higher-Order Map}\label{sec:ctrlmap}
In quantum circuits, quantum control is realised through controlled gates.
Before delving into the main proofs, in this section, we illustrate the correspondence between the controlled diagrams investigated in this work and the conventional concept of controlled gates in quantum circuits.
Specifically, we can derive the latter as a special case of the former.
We show this by reasoning with controlled gates through a straightforward construction on top of our ring of controlled square matrices.
We define the higher-order map $\Control$ which takes a square matrix $M: V \to V$ to its controlled square diagram $V \otimes \mathbb{C}^2 \to V \otimes \mathbb{C}^2$. In the functorial box notation of \cite{mellies2006functorial}, we write
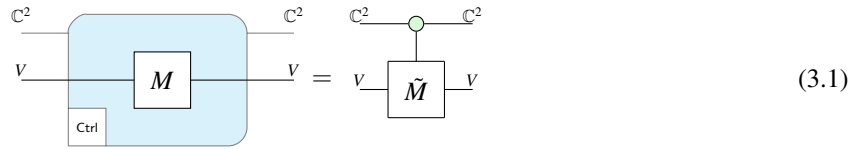
\begin{equation}
    \input{tikz/func/F_def_box.tikz}
\end{equation}
where
\begin{equation}
    \left\llbracket \quad \input{tikz/con/c_sq_str8.tikz} \;\;\; \right\rrbracket ~=~ \begin{bmatrix}I & 0 \\ 0 & M\end{bmatrix}
\end{equation}
Detailed exploration of control as a functor, extending props to controlled props, was independently carried out by Delorme and Perdrix in~\cite{delorme2026ctrl}.

For example, the CNOT gate can be defined from a controlled $X$ gate since:
\begin{equation}
    \input{tikz/con/Fbox_cx.tikz}
\end{equation}

We prove in Appendix~\ref{sec:ctrlmapproofs} that composition of controlled operations in sequence and in parallel is well-behaved.

\begin{prop}\label{prop:ctrl_comp_h}
\begin{equation}
	\input{tikz/func/ctrl_comp_hdef.tikz}
\end{equation}\end{prop}

\begin{prop}\label{prop:ctrl_comp_v}
\begin{equation}
	\input{tikz/func/ctrl_comp_vdef.tikz}
\end{equation}
\end{prop}

Furthermore, successive applications of $\Control$ recovers the standard notion of multiple-control, which computes the AND of the control qubits:
\begin{prop}\label{prop:FF}
    \begin{equation}
        \input{tikz/func/FF_statement.tikz}
    \end{equation}
\end{prop}
\section{Ring Axioms for Controlled Diagrams}\label{sec:ring}
In this section, we reverse-engineer the underlying algebraic properties of controlled state and controlled square matrix diagrams. This builds up to diagrams for the unique normal form for states used for the first proofs of complete axiomatisation for qubit graphical calculi~\cite{hadzihasanovic2017thesis, Hadzihasanovic2018zwzxcomplete}.
All proofs in this section can be found in Appendix~\ref{sec:ringproofs}.

There are only five Controlled Ring Rules necessary to derive the new completeness results in this work, presented in Figure~\ref{fig:zxw_rules}.
All other diagrammatic equalities in this paper were proven using only the rules in Figure~\ref{fig:zxw_rules}.
We proved these five rules through deducing action on basis states, either from our definitions of controlled states and controlled square matrices, or through Lemmas~\ref{lem:csq_copy}, \ref{lem:csqaddcomm}, and \ref{lem:cscopy}.


Let $\tilde{M_n}$ be the set of controlled square matrices on $n$ qubits. The goal of this section is to prove that the addition and multiplication operations introduced above induce a ring on $\tilde{M_n}$. By Proposition \ref{prop:cmat_ops}, the addition and multiplication of controlled matrices is just the controlled addition and multiplication of the underlying matrices so the fact that the ring properties hold is not particularly surprising. What is more interesting is how easily these properties can be proven with a small subset of the ZXW rules. Likewise, we show that the set of controlled $n$-qubit states $\tilde{S_n}$ also forms a ring. The first lemma enables us to copy controlled matrices.

\begin{lemma}\label{lem:csq_copy}
    For any square matrix $M$, 
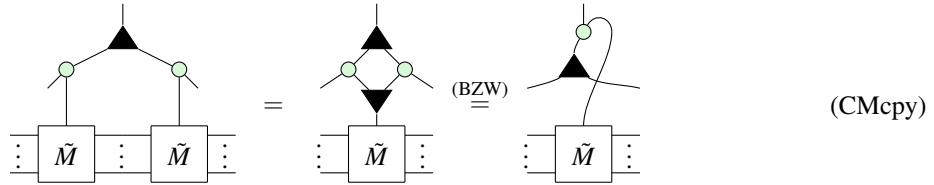
    \begin{equation*}
    \input{tikz/con/csq_copy_bzw.tikz}
    \tag{\ref{rule:CMcpy}}
    \end{equation*}
\end{lemma}

Now we show that controlled matrix addition and multiplication satisfy the ring axioms. Associativity of $+, \times$ follow immediately from (\ref{rule:Aso}, \ref{rule:S1}), respectively. Commutativity of addition follows from the commutativity of matrix addition and Proposition \ref{prop:cmat_ops}, defining the following rewrite rule for controlled operators acting on mutually exclusive sectors:

\begin{lemma}\label{lem:csqaddcomm}
    Let $M_1, M_2$ be $n \times n$ matrices. 
    \begin{equation*}
       \input{tikz/con/csq_add_comm_statement.tikz}
       \tag{\ref{rule:CMcom}}
    \end{equation*}
\end{lemma}

\begin{lemma}\label{lem:csqaddid}
The additive identity is defined as $\redzeroup \otimes I_n$:
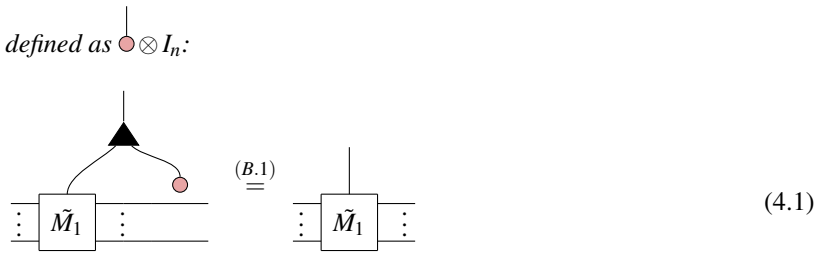
\begin{equation}
    \input{tikz/con/csq_add_id.tikz}
\end{equation}
\end{lemma}

The multiplicative identity is defined very similarly as $\lowerbox{\greenzeroup} \otimes I_n$. The existence of additive inverses relies on the copying lemma from before.

\begin{lemma}\label{lem:csq_inv}
    The additive inverse of $\tilde{M}$ is $\raisebox{-5pt}{\numbergate[-1]} \circ \tilde{M}$.
\end{lemma}

\begin{lemma}\label{lem:csq_dist}
    The addition and multiplication operations of controlled matrices distribute:
    \begin{equation}
        \input{tikz/con/csq_dist_statement.tikz}
    \end{equation}
\end{lemma}

Combining the lemmas of this section shows that controlled matrices form a ring. A similar result can be shown for controlled states. Once again, we start with the ability to copy controlled states. 
\begin{lemma}\label{lem:cscopy}
    For any state $\psi$,
    \begin{equation*}
        \input{tikz/con/cs_copy.tikz}
        \tag{\ref{rule:CScpy}}
    \end{equation*}
\end{lemma}

Several of the ring axioms follow directly from basic ZXW rules, such as commutativity of addition:

\begin{lemma}\label{lem:csaddcomm}
    For $n$-partite states $\psi_1, \psi_2$, $\tilde{\psi_1} \boxplus \tilde{\psi_2} = \tilde{\psi_2} \boxplus \tilde{\psi_1}  $.
\end{lemma}

Associativity of $\boxplus$ follows similarly, using (\ref{rule:Aso}). Next we have the additive identity.

\begin{lemma}\label{lem:csaddid}
    $\tilde{\psi} \boxplus\mathbf{\tilde{0}} = \tilde{\psi}$.
\end{lemma}

The additive inverse is defined similarly to the case of controlled matrices.

\begin{lemma}\label{lem:csaddinv}
    For a controlled state $\tilde{\psi}$, its additive inverse is $\tilde{\psi} \circ \raisebox{-7pt}{\numbergate[-1]}$.
\end{lemma}

Associativity and commutativity of $\boxtimes$ follow as before, using (\ref{rule:S1}) for $\lowerbox{\zspid}$. Finally, we must prove distributivity.

\begin{lemma}\label{lem:csdist}
    $\tilde{\psi_1} \boxtimes (\tilde{\psi_2} \boxplus \tilde{\psi_3}) = (\tilde{\psi_1} \boxtimes \tilde{\psi_2}) \boxplus (\tilde{\psi_1} \boxtimes \tilde{\psi_3})$.
\end{lemma}

\begin{remark}\label{remark:entrywise}
    A different addition and multiplication for controlled states was defined in Ref.~\cite{jeandel2018zxconstructive}. There corresponded to entry-wise addition and multiplication of statevectors, while our $\boxplus$ and $\boxtimes$ correspond to addition and multiplication of polynomials in bijective correspondence to controlled states, which we show next.
\end{remark}

\section{Isomorphism between the Ring of Controlled States and Multilinear Polynomials}\label{sec:iso}

Its been known since 2011 that $\lowerbox{\wspid}, \lowerbox{\zspid}$ can be used to add and multiply number states $\lowerbox{\numberstate}$, respectively \cite{coecke2011ghz}. In the previous section we saw that $\lowerbox{\wspid}, \lowerbox{\zspid}$ can moreover be used to copy controlled diagrams. In this section, we explain this connection by demonstrating that controlled states are in fact isomorphic to multilinear polynomials. This being a bijection is a well-known folklore result in the study of entangled states, but to the best of our inquiries we are not aware of a proof. More generally, Ref.~\cite{wilson2023diffpolycirc} presented Cartesian Distributive Categories exemplified by polynomial circuits, which are isomorphic to polynomials over arbitrary commutative semirings or rings; their proof is non-constructive, giving explicit proof only for the case of Boolean circuits~\cite{wilson2021revderbool}. Our proof hinges on a normal form inspired by the recent proof of completeness for the ZXW calculus \cite{poor2023completeness}, suggesting that much of the expressive power of the ZXW calculus comes from this algebraic structure.

Firstly, we describe how to interpret certain ZXW diagrams as polynomials. Consider the diagrams:
\begin{equation}
    \input{tikz/poly/eg1.tikz}
\end{equation}

If we treat the bottom wires as an indeterminate $x$, we can read these bottom-up as computing $x - 1$ and $2x + 3$, respectively. Moreover, since these diagrams are both controlled states, they can be added together,  yield a diagram resembling $3x + 2$:
\begin{equation}\label{eq:eg_add}
    \input{tikz/poly/eg4.tikz}
\end{equation}

When trying to multiply these diagrams, rather than getting $(x-1)(2x+3) = 2x^2 + x - 3$, we instead get $x - 3$.
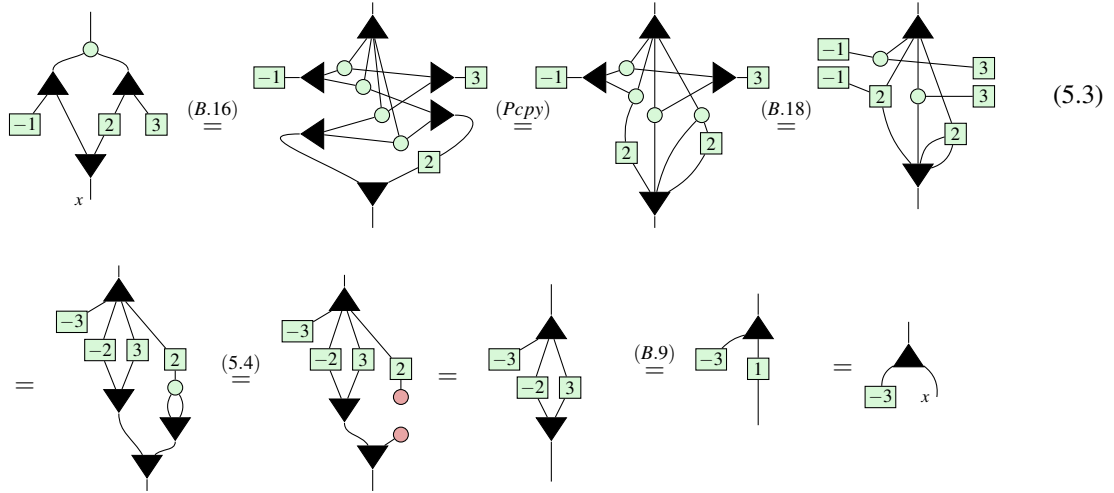
\begin{equation}\label{eq:eg_times}
    \input{tikz/poly/eg5.tikz}
\end{equation}

The reason for the missing $2x^2$ term is due to the following:
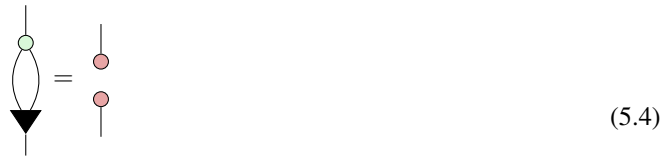
\begin{lemma}\label{lem:kill_quad}
    \begin{equation}\label{eq:kill_quad}
        \input{tikz/lemmas/kill_quad_statement.tikz}
    \end{equation}
\end{lemma}
This lemma, proven in Appendix~\ref{sec:applem}, implies that $x^2 = 0$ for all variables $x$. Other than this, controlled state arithmetic appears to faithfully reflect polynomial arithmetic. To formalise this correspondence, we introduce the following definition.

\begin{definition}\label{def:arithmetic}
    A ZXW diagram with a single input on top is \textbf{arithmetic} if it contains only  $\lowerbox{\idwire}$, $\lowerbox{\swap}$ wires, $\lowerbox[10]{\wspids}$, $\lowerbox[10]{\zspids}$, $\lowerbox{\coWs}$ nodes and $\lowerbox{\numberstate}$ boxes.
\end{definition}
\begin{remark}
    This fragment of the ZXW-calculus defines a subcategory; adding $\lowerbox{\redzero}$, this is an instance of a Cartesian Distributive Category as defined in Ref.~\cite{wilson2023diffpolycirc}.
\end{remark}

To interpret an arithmetic ZXW diagram as an arithmetic expression, read $\lowerbox{\wspid}$ as $+$, $\lowerbox{\zspid}$ as $\times$, $\lowerbox{\numberstate}$ as the number $a$, $\lowerbox{\coW}$ as fanout and output/bottom wires as variables $x_1, ..., x_n$ numbered from left to right. The following lemma establishes that all arithmetic diagrams are controlled states:
\begin{lemma}
    For any arithmetic diagram $A$, \begin{equation}\label{eq:arith_cs}
        \input{tikz/poly/arith0_statement.tikz}
    \end{equation}
\end{lemma}
\begin{proof}
    By definition, other than wires $A$ contains only $\raisebox{-10pt}{\wspids}$, $\raisebox{-10pt}{\zspids}$, $\raisebox{-5pt}{\coWs}$, and $\raisebox{-3pt}{\numberstate}$. All $\raisebox{-3pt}{\numberstate}$'s can be removed with (\ref{rule:Ept}). Meanwhile all the spiders copy $\raisebox{-5pt}{\redzero}$ due to (\ref{rule:Bs0}, \ref{rule:K0}, \ref{eq:wid}) respectively.
\end{proof}

Just as it is typical to represent a polynomial in normal form as a sum of products, it is possible to rewrite every arithmetic diagram into a normal form as a single $\lowerbox[10]{\wspids}$, followed by a layer of $\lowerbox[10]{\zspids}$, followed by a layer of $\lowerbox[2]{\numberstate}, \lowerbox[2]{\coWs}$. 

\begin{definition}
    An $n$-output arithmetic diagram is said to be written in \textbf{polynormal form} (PNF) if it is of the form:
    \begin{equation}\label{eq:pnf}
        \input{tikz/poly/pnf.tikz}
    \end{equation}

    The $i$th coefficient $a_i$ is connected to the $k$th $\lowerbox{\coWs}$ iff the $k$th bit in the binary expansion of $i$ is 1. 
\end{definition}

This normal form is familiar from completeness of all linear maps for qubits~\cite{Hadzihasanovic2018zwzxcomplete} and for qudits~\cite{poor2023completeness}. The reason we introduce the definition of a PNF is that it is an arithmetic diagram and therefore has a more immediate arithmetic interpretation. The reason for the specific connectivity condition is that it enables a diagram in PNF to directly represent its own matrix.

\begin{prop}\label{prop:vec_pnf}
The diagram in Equation~\eqref{eq:pnf} equals the matrix
    $\begin{bmatrix}
            1 &  a_0 \\ 0 & a_1 \\ \vdots & \vdots \\ 0 & a_{2^n-1}
        \end{bmatrix}$
\end{prop} 

\begin{proof}
    See Appendix~\ref{sec:isoproofs}.
\end{proof}

Thus, every controlled state can be represented as at least one arithmetic diagram (namely, its PNF). Moreover, we now show that any other arithmetic diagram can always be rewritten to its PNF.

\begin{thm}\label{thm:uni_pnf}
    All arithmetic diagrams can be rewritten into PNF through application of ZXW-calculus rules. Therefore, the rules applied suffice for completeness for the arithmetic fragment of the ZXW-calculus.
\end{thm}

\begin{proof}
    We present an algorithm to rewrite any arithmetic diagram to PNF in Appendix~\ref{sec:isoproofs}.
\end{proof}
\begin{remark}  
    All of the qubit ZXW-calculus rules~\cite{poor2023completeness} applied in this proof are derivable from the minimal qubit ZW-calculus rules from~\cite{deVisme2025minzw}, derived independently from this work~\cite{Agnew2023Masters}. We can replace the \ref{rule:TA} rule with the rule of Lemma~\ref{lem:kill_quad}, and the qubit ZXW-calculus \ref{rule:WW} rule is derived in their~\cite[Lemma 23]{deVisme2025minzw}.
\end{remark}

\subsection{Isomorphism}
At last we can prove the isomorphism. Recall that $\tilde{S_n}$ is the ring of controlled states with $n$ outputs. Throughout, we shall let $\polyring$ denote the multilinear ring $\mathbb{C}[x_1, ..., x_{n}]/(x_1^2, ..., x_{n}^2)$.

\begin{thm}\label{thm:iso}
    There is an isomorphism $\polyring \simeq \tilde{S_n}$
\end{thm}

First, we shall define the map $\phi_n: \polyring \to \csring$ before proving it induces an isomorphism. $\phi_n$ is defined to map an arbitrary polynomial $p(x_1, ..., x_n) = a_0 + a_1x_n + ... + a_{2^n-1}x_1x_2...x_n$ to the PNF in equation (\ref{eq:pnf}).

Some important special cases are mapping scalars $a \in \mathbb{C}$ and indeterminates $x_i$:
 \begin{equation}
        \input{tikz/poly/homproof/homdef1.tikz} \quad \input{tikz/poly/homproof/homdef2.tikz}
\end{equation}
   
The proof that $\phi_n$ is a homomorphism resembles equations (\ref{eq:eg_add}) and (\ref{eq:eg_times}), but with greater generality. That $\phi_n$ is a bijection relies on Proposition \ref{prop:vec_pnf}. The full proof is found in Appendix~\ref{sec:isoproofs}.

\section{Completeness for Factoring Controlled Operators}
Instead of the indeterminates being complex numbers represented by $\lowerbox{\numberstate}$'s, we can let them be same-size matrices represented by controlled square matrix diagrams.
We then have that:
\begin{thm}\label{thm:ctrl_pnf}
    ZXW diagrams where the outputs of an arithmetic ZW diagram are each plugged into controls of same-size controlled matrices, are isomorphic to multivariate polynomials over same-size square matrices with complex number coefficients.
    The rules for their completeness are the same subset of ZXW rules used for completeness for arithmetic diagrams in the ZXW-calculus in \Cref{thm:uni_pnf}, plus the \ref{rule:CMcpy} and \ref{rule:CMcom} rules.
    When using the minimal set of rules from~\cite{deVisme2025minzw} for completeness for arithmetic diagrams, adding the \ref{rule:CMcpy} and \ref{rule:CMcom} rules is also a minimal rule set, i.e. none of the rules are derivable from the others.
\end{thm}
\begin{proof}
    The proof is by deriving a unique normal form, having fixed an (arbitrary) order on the same-size square matrix variables.
    The algorithm to arrive at this normal form starts by rewriting the diagram sans the controlled square matrices to PNF, by the procedure of \Cref{thm:uni_pnf}.
    Next, remove all $\lowerbox{\coWs}$'s by using \eqref{rule:CMcpy} to make copies of the controlled square matrices.
    Finally, use \eqref{rule:CMcom} to commute controlled square matrices whose controls act on mutually exclusive sectors past each other.
    The algorithm terminates with the controlled square matrix terms ordered by their mutually exclusive sectors in lexicographical order with respect to the order of their variables.

    To show minimality, we can consider that the \ref{rule:CMcpy} and \ref{rule:CMcom} rules are the only rules here that involve the controlled square matrix generator, and so they cannot be derivable from the other rules here.
    The \ref{rule:CMcpy} rule can change the total number of controlled square matrices in the diagram, while the \ref{rule:CMcom} rule cannot.
    On the other hand, the \ref{rule:CMcom} rule can apply to two different controlled square matrices, changing their order where the outputs of one are inputs to the other; the \ref{rule:CMcpy} rule applies to two controlled square matrices only if they are of the same matrix.
\end{proof}
\begin{remark}
    These procedures guarantee that controlled operators in mutually exclusive sectors are commutable past each other.
    This guarantee does not apply to controlled operators in the same sector, as this would require additional information about the operators themselves beyond being controlled operators.
\end{remark}

\subsection{Factorising Hamiltonians}
To present a small working example, we leverage both our rewrite rules for arithmetic ZW diagrams, and for controlled diagrams, to \emph{factor} them.  For example, for same-size square matrices $I, A, B$ and $a, b, c \in \mathbb{C}$:
\begin{gather}
    \input{tikz/poly/matspoly.tikz}
\end{gather}
Factoring Hamiltonians is important to optimise quantum algorithms for chemistry and physics simulations. However, previous graphical rewrites for factoring Hamiltonians had only been doable for Hamiltonians with concretely-specified matrix terms~\cite{shaikh2022sum}. This completeness result guarantees that for any Hamiltonian, even if its matrix terms are black-box, these graphical rewrite rules are capable of deriving any of its possible factorisations.

Recently, one of the new rules for controlled diagrams first proposed in this work, \eqref{rule:CMcpy}, was used in Ref.~\cite{mcdowallrose2025f2q} to derive a general form for any two-body fermionic Hamiltonian in second quantisation under any linear encoding.

\section{Conclusion}
To conclude, we proved completeness for all controlled $n$-partite states, which we showed form a commutative ring isomorphic to multilinear polynomials. Also, we showed that all controlled $n$-qubit square matrices form a non-commutative ring. Furthermore, we have completeness for plugging controlled states into the control wires of controlled diagrams, isomorphic to all multivariate polynomials over same-size square matrices, with application to factoring Hamiltonians. When the controls target mutually exclusive sectors, a rewrite rule can be applied to copy any controlled diagram, and thereby factor any Hamiltonian.

We have shown that every (controlled) state computes a polynomial; hence, we can interpret a universal fragment of the ZXW-calculus as corresponding to arithmetic circuits. This generalises Ref.~\cite{carette2023compositionality}, which found an algebraic interpretation of a certain fragment of ZW calculus.
This line of reinterpreting quantum circuits as computing polynomials rather than unitary matrices offers a new perspective on quantum computation.

In another direction, we can apply completeness for polynomials isomorphic to controlled states to study entanglement. It can be easily shown diagrammatically that the polynomials corresponding to entangled (non-separable) states are exactly those that cannot be factored into irreducibles containing only variables corresponding to Alice's subsystem or only corresponding to Bob's subsystem. Since there are efficient algorithms for polynomial factorisation ~\cite{forbes2015complexity}, this gives rise to a novel entanglement classification algorithm for pure states. Further developing this into a more refined algebraic theory of entanglement, building on the work in Ref.~\cite{Agnew2023Masters}, could offer further insights.

The natural next step is extend our results to controlled \textit{qudit} diagrams.
While the diagrams being controlled are over qudits, we can consider control in the qubit subspace, as done in the ZXW-calculus completeness proof for any qudit dimension~\cite{poor2023completeness}.
A starting guess would be that qudit controlled states are isomorphic to polynomials $\mathbb{C}^{d-1}[x_1,...,x_n]/({x_1}^d,...,{x_n}^d)$ due to the Hopf law between Z and W.
Qudit multiple-control would likely have more complex structure than the qubit case here, considering the constructions for all prime-dimensional $d$-ary classical reversible gates and W states built in Refs.~\cite{Yeh2022QutritCtrl, Roy2023quditzh, Yeh2023quditW}.

Last, we are interested in exploring how to embed these new semantics for quantum controlled states and matrices into a functional programming language like in Ref.~\cite{rennela2020clctrllinlogic}, and to interpret our diagrams in mixed-state quantum mechanics and relate them to the operational and denotational semantics for coherent control of Ref.~\cite{barsse2025qplcoherent} and to the channel construction of Ref.~\cite{deFelice2026dataflow}.
We would like to try sector-preserving channels~\cite{Vanrietvelde2021ctrlsector} and scoped effects~\cite{lindley2024scoped} as approaches to better formulate the semantics of multiple-control.
We are also curious about reconciling the interpretation of diagrammatic differentiation of our arithmetic polynomial circuits by the approach in Ref.~\cite{wilson2023diffpolycirc}, with that of quantum circuits and ZX diagrams in Refs.~\cite{toumi2021diagdiff, wang2022diffintzx, jeandel2024adddiffzx}.

\section{Acknowledgements}
We thank Bob Coecke for motivating our exploration of these topics. We are grateful to all the collaboration with Matt Wilson in discovering the functorial properties of control. We thank Razin Shaikh, Harny (Quanlong) Wang, Itai Leigh, Tim Forrer, Aleks Kissinger, Frank (Peng) Fu, and Maxim van den Berg for insightful discussions. We also thank the anonymous reviewers who have given feedback on this work for their helpful comments.
LY would like to thank the Google PhD Fellowship, and the Basil Reeve Graduate Clarendon Scholarship at Oriel College, for PhD funding.
RY would like to thank Simon Harrison for his generous support via the Wolfson Harrison UKRI Quantum Foundation Scholarship.

\newpage

\newpage
\appendices
\section{Proofs for Section~\ref{sec:ctrlmap}}\label{sec:ctrlmapproofs}

\textbf{Proof of Proposition \ref{prop:ctrl_comp_h}}
\begin{proof}
    $\Control$ is sound with regards to sequential composition:
    \begin{equation}
        \input{tikz/func/F_comp.tikz}
    \end{equation}

    Where $(*)$ follows from Proposition \ref{prop:cmat_ops}.
\end{proof}

\textbf{Proof of Proposition \ref{prop:ctrl_comp_v}}
\begin{proof}
    Consider the morphism: $\phi_{V, W}: \Control(V) \otimes \Control(W) \to \Control(V \otimes W)$:
    \begin{equation}
       \phi_{V, W} = \input{tikz/func/F_lax_def.tikz}
    \end{equation}

    By conjugating by $\phi$, $\Control$ is sound under parallel composition of any $M_1: V \to V,\; M_2: W \to W$:
    \begin{equation}
        \input{tikz/func/ctrl_comp.tikz}
    \end{equation}

    Parallel composition is associative by associativity of the tensor product and of $\phi$:
    \begin{equation}
       \input{tikz/func/F_lax_ass.tikz}
    \end{equation}
\end{proof}

\textbf{Proof of Proposition \ref{prop:FF}}
\begin{proof}
    Plugging basis states:
    \begin{equation}
        \input{tikz/func/FF1.tikz}
    \end{equation}
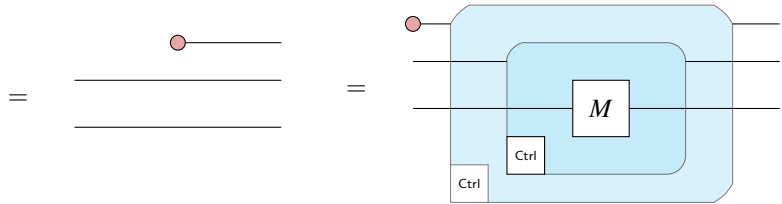
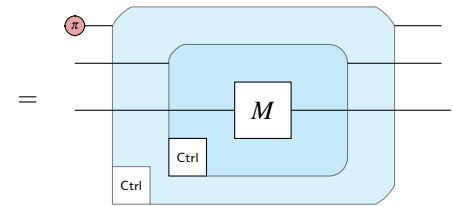
    \begin{equation}
        \input{tikz/func/FF2.tikz}
    \end{equation}
\end{proof}

The AND gate is a primitive in the ZH-calculus~\cite{backens2019zh}. We can also represent the binary AND gate in the ZXW-calculus, by deMorgan's law of the diagram for binary OR given in \Cref{prop:vec_pnf}.
\begin{lemma}\label{lemma:and}
    \begin{equation}
        \input{tikz/func/and_def.tikz}
    \end{equation}
\end{lemma}
\begin{proof}
    We can verify this computes the AND gate by computing on basis states.
    \begin{equation}\label{eq:and1}
        \input{tikz/func/and1.tikz}
    \end{equation}
    Thus $AND(1, x) = x$. Since the diagram is clearly commutative, it remains to check $AND(0, x) = 0$.
    \begin{equation}\label{eq:and0}
        \input{tikz/func/and0.tikz}
    \end{equation}
\end{proof}

\section{Basic Lemmas}\label{sec:applem}
\begin{lemma}
    As a special case of~\cite[Lemma 38]{poor2023completeness} for the qubit setting, we have:
    \begin{equation}\label{eq:wid}
        \input{tikz/lemmas/wid.tikz}
    \end{equation}
\end{lemma}

\begin{lemma}
    The below follows from~\cite[Lemma 37]{poor2023completeness}, $\ref{rule:K0}$, and~\cite[Lemma 3.23]{wang2022algebraiczx}:
    \begin{equation}\label{eq:wont}
        \input{tikz/lemmas/wont.tikz}
    \end{equation}
\end{lemma}

\begin{lemma}
\begin{equation}\label{eq:inner}
	\input{tikz/lemmas/inner_statement.tikz}
\end{equation}
\end{lemma}
\begin{proof}
\begin{equation}
\input{tikz/lemmas/inner.tikz}
\end{equation}
\end{proof}

\begin{lemma}
    \begin{equation}\label{eq:xcpy}
        \input{tikz/lemmas/xcpy_statement.tikz}
    \end{equation}
    \end{lemma}
    \begin{proof}
        \begin{equation}
        \input{tikz/lemmas/xcpy.tikz}
    \end{equation}
    \end{proof}

  \begin{lemma}
      \begin{equation}\label{eq:zerobox}
          \input{tikz/lemmas/zerobox_statement.tikz}
      \end{equation}
  \end{lemma}
  \begin{proof}
      \begin{equation}
          \input{tikz/lemmas/zerobox.tikz}
      \end{equation}
  \end{proof}
  
  \begin{lemma}
      \begin{equation}\label{eq:cp_add}
          \input{tikz/lemmas/cpadd_statement.tikz}
      \end{equation}
  \end{lemma}
  \begin{proof}
      \begin{equation}
      \input{tikz/lemmas/cpadd.tikz}
  \end{equation}
  \end{proof}
  
  \begin{lemma}
    \begin{equation}\label{eq:x-x=0}
      \input{tikz/lemmas/cp_add2_statement.tikz}
    \end{equation}
  \end{lemma}
  \begin{proof}
    \begin{equation}
        \input{tikz/lemmas/cp_add2.tikz}
    \end{equation}
  \end{proof}

  \begin{lemma}
    \begin{equation}\label{eq:cpk_add}
      \input{tikz/lemmas/cpk_add_statement.tikz}
    \end{equation}
  \end{lemma}
  \begin{proof}
    \begin{equation}
      \input{tikz/lemmas/cpk_add.tikz}
    \end{equation}
  \end{proof}
  
\textbf{Proof of Lemma \ref{lem:kill_quad}}
\begin{proof}
\begin{equation}
    \input{tikz/lemmas/kill_quad.tikz}
\end{equation}
\end{proof}  
  
  \begin{lemma}
    \begin{equation}\label{eq:dbl_dist}
    \input{tikz/lemmas/dbl-dist-statement.tikz}
  \end{equation}
  \end{lemma}
  \begin{proof}
    \begin{equation}
    \input{tikz/lemmas/dbl-dist.tikz}
  \end{equation}
  \end{proof}
  
  \begin{lemma}
    \begin{equation}\label{eq:dist_circ}
      \input{tikz/lemmas/dist_statement.tikz}
    \end{equation}
  \end{lemma}
  \begin{proof}
    \begin{equation}
      \input{tikz/lemmas/dist.tikz}
    \end{equation}
  \end{proof}

  


  \begin{lemma}
    \begin{equation}\label{eq:0times}
        \input{tikz/lemmas/0times_statement.tikz}
    \end{equation}
  \end{lemma}
  \begin{proof}
    \begin{equation}
        \input{tikz/lemmas/0times.tikz}
    \end{equation}
  \end{proof}

\section{Proofs for Section~\ref{sec:ring}}\label{sec:ringproofs}
\textbf{Proof of Lemma \ref{lem:csq_copy}}
\begin{proof}
    First of all, using (\ref{rule:BZW}) we can rewrite the LHS to
    \begin{equation}
        \input{tikz/con/csq_dcopy2.tikz}
    \end{equation}

    Then clearly 
    \begin{equation}
        \input{tikz/con/csq_dcopy3.tikz}
    \end{equation}

    Meanwhile, 
    \begin{equation}
        \input{tikz/con/csq_dcopy4.tikz}
    \end{equation}

    Thus the two sides are equal over the Z basis and so are equal as diagrams.
\end{proof}

\textbf{Proof of Lemma \ref{lem:csqaddcomm}}
\begin{proof}
    We prove by plugging red and commutativity of matrix addition. By definition of controlled matrices, plugging $\lowerbox{\redzero}$ gives $\lowerbox{\redzero}$ and $I_n$ on both sides. Meanwhile, plugging $\lowerbox{\redpi}$ gives:
    \begin{equation}
       \input{tikz/con/csq_add_comm.tikz}
   \end{equation}
\end{proof}

\textbf{Proof of Lemma \ref{lem:csq_inv}}
\begin{proof}
    \begin{equation}
        \input{tikz/con/csq_add_inv.tikz}
    \end{equation}
\end{proof}

\textbf{Proof of Lemma \ref{lem:csq_dist}}
\begin{proof}
    \begin{equation}
        \input{tikz/con/csq_dist.tikz}
    \end{equation}
\end{proof}

\textbf{Proof of Lemma \ref{lem:cscopy}}
\begin{proof}
    As before, plugging $|0\rangle$ gives
    \begin{equation}
        \input{tikz/con/cs_copy1.tikz}
    \end{equation}

    Meanwhile, plugging $|1\rangle$ gives
        \begin{equation}
        \input{tikz/con/cs_copy2.tikz}
    \end{equation}
\end{proof}

\textbf{Proof of Lemma \ref{lem:csaddcomm}}
\begin{proof}
    \input{tikz/con/cs_add_comm.tikz}
\end{proof}

\textbf{Proof of Lemma \ref{lem:csaddid}}
\begin{proof}
    It is clear that $\lowerbox[10]{\zeroproj}$ is the controlled state $\tilde{\mathbf{0}}$. 
    
    Then we have:
    \begin{equation}
        \input{tikz/con/cs_add_id.tikz}
    \end{equation}

\end{proof}

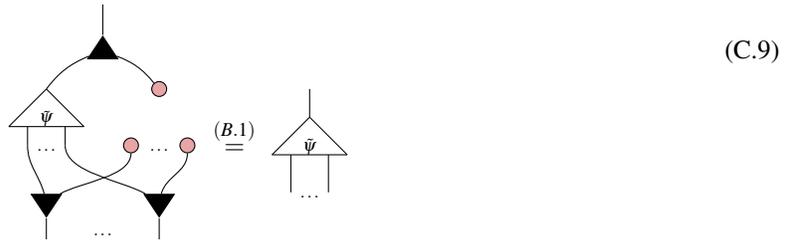
\textbf{Proof of Lemma \ref{lem:csaddinv}}
\begin{proof}
    $\tilde{\psi} \circ \raisebox{-5pt}{\numbergate[-1]}$ is still a controlled state since $\raisebox{-5pt}{\numbergate[-1]}$ does nothing to $\raisebox{-5pt}{\redzero}$. Then $\tilde{\psi} \circ \raisebox{-5pt}{\numbergate[-1]}$ inverts $\tilde{\psi}$ since:
    \begin{equation}
        \input{tikz/con/cs_add_inv.tikz}
    \end{equation}
\end{proof}

\textbf{Proof of Lemma \ref{lem:csdist}}
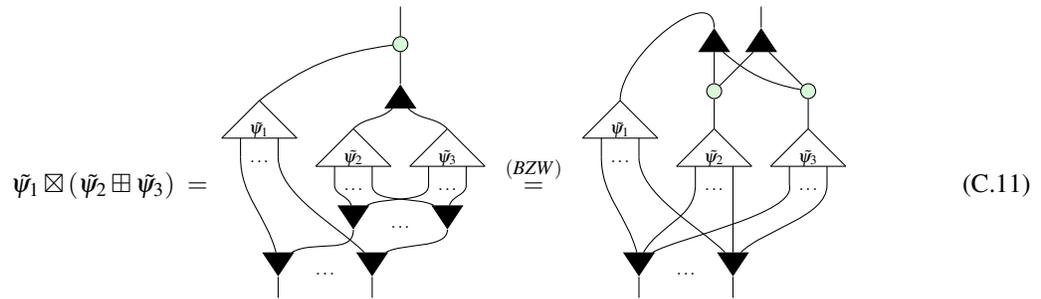
\begin{proof}
    \begin{gather} 
        \tilde{\psi_1} \boxtimes (\tilde{\psi_2} \boxplus \tilde{\psi_3}) ~=~ \input{tikz/distproof/d1.tikz} ~
        \eqq{\ref{rule:BZW}} ~ \input{tikz/distproof/d2.tikz} \neweqline ~=~ \input{tikz/distproof/d3.tikz} ~
        \eqq{\ref{rule:CScpy}} ~ \input{tikz/distproof/d4.tikz} \neweqline ~=~ \input{tikz/distproof/d5.tikz}
        ~=~ \input{tikz/distproof/d6.tikz} \neweqline ~=~ (\tilde{\psi_1} \boxtimes \tilde{\psi_2}) \boxplus (\tilde{\psi_1} \boxtimes \tilde{\psi_3})
    \end{gather} 

\end{proof}

\section{Proofs for Section~\ref{sec:iso}}\label{sec:isoproofs}
\textbf{Proof of Proposition \ref{prop:vec_pnf}}
\begin{proof}
    We prove by induction on $n$.
    
    For the base case, $n=0$. The only PNF with no outputs is a number so we have: $$\numberstate[a_0] = \begin{bmatrix}
        1 & a_0
    \end{bmatrix}$$ as desired.
    
    For inductive hypothesis, we assume that Proposition~\ref{prop:vec_pnf} holds for every PNF on $n$ outputs. We use this hypothesis to extend it to PNFs with $n+1$ outputs. 
    
    Let $D$ be an arbitrary PNF with $n+1$ outputs. Firstly, observe that $x_{n+1}$ is connected to only the odd coefficients $\{a_{2k+1}\}$ since these are exactly the indices with $1$ in the least significant bit. Thus we can rewrite:
    \begin{gather}
        \raisebox{-20pt}{\ddiag} ~=~ \input{tikz/poly/lemmas/uni10.tikz} \neweqline ~=~ \input{tikz/poly/lemmas/uni11.tikz} \neweqline ~\eqq{\ref{rule:BZW}}~ \input{tikz/poly/lemmas/uni12.tikz} \neweqline 
        ~=~ \input{tikz/poly/lemmas/uni13.tikz}
    \end{gather}
    Where $D_{even}, D_{odd}$ are PNF diagrams. Since they are over $n$ variables, we can apply the inductive hypothesis and obtain:
    \begin{equation}\tag{*}
        D_{even} = \begin{bmatrix}
            1 & a_0 \\ 0 & a_2 \\ \vdots & \vdots \\ 0 & a_{2^{n+1}-2}
        \end{bmatrix}, \qquad \qquad
        D_{odd} = \begin{bmatrix}
            1 & a_1 \\ 0 & a_3 \\ \vdots & \vdots \\ 0 & a_{2^{n+1}-1}
        \end{bmatrix}
    \end{equation}

    Next, plugging red we observe:
    \begin{equation}
        \input{tikz/poly/lemmas/uni2.tikz}
    \end{equation}
    Meanwhile,
    \begin{equation}
        \input{tikz/poly/lemmas/uni3.tikz}
    \end{equation}

    Summing these together,
    \begin{gather}
        \input{tikz/poly/lemmas/uni4.tikz}
        \neweqline ~=~ (D_{even} \otimes \ketz) + (D_{odd}\keto \brao \otimes  \keto) \neweqline \eqq{*}~ \begin{bmatrix}
            1 & a_0 \\ 0 & a_2 \\ \vdots & \vdots \\ 0 & a_{2^{n+1}-2}
        \end{bmatrix} \otimes \ketz + \begin{bmatrix}
            0 & a_1 \\ 0 & a_3 \\ \vdots & \vdots \\ 0 & a_{2^{n+1}-1}
        \end{bmatrix} \otimes \keto  \neweqline
        ~=~ \begin{bmatrix}
            1 & a_0 \\ 0 & 0 \\ 0 & a_2 \\ 0 & 0 \\ \vdots & \vdots \\ 0 & a_{2^{n+1}-2} \\ 0 & 0
        \end{bmatrix} + \begin{bmatrix}
            0 & 0 \\ 0 & a_1 \\ 0 & 0 \\ 0 & a_3 \\ \vdots & \vdots \\ 0 & 0 \\ 0 & a_{2^{n+1}-1}
        \end{bmatrix} 
        ~=~ \begin{bmatrix}
            1 & a_0 \\ 0 & a_1 \\ 0 & a_2 \\ 0 & a_3 \\ \vdots & \vdots \\ 0 & a_{2^{n+1}-2} \\ 0 & a_{2^{n+1} -1}
        \end{bmatrix}
    \end{gather}

    Completing the inductive step.
\end{proof}

\textbf{Proof of Theorem \ref{thm:uni_pnf}}

\begin{proof}
    Let $A$ be an arithmetic diagram. If $A = \numberstate$, we are done. 
    
    Otherwise, $A$ has at least one output. First, we shall rewrite $A$ into three layers, consisting of: (1) a single W at the top, (2) a layer of $\raisebox{-5pt}{\zspids}$ and (3) a layer of $\numberstate$'s and $\raisebox{-5pt}{\coWs}$'s. Then we shall collect terms and order the boxes to produce a PNF. 

    If the top of $A$ is not already $\raisebox{-5pt}{\wspids}$, it must be $\raisebox{-5pt}{\zspids}$. It cannot be $\numberstate$ since the remaining arithmetic diagram would then have no inputs which is impossible. It cannot be $\raisebox{-5pt}{\coWs}$ since there is only one input and arithmetic diagrams cannot contain $\ccap$. Thus we can rewrite:
    \begin{enumerate}[label={(\arabic*)}]
    \item $\input{tikz/poly/lemmas/algtop1.tikz}$
    \end{enumerate}

    (1) guarantees there is a W at the top. We shall now repeatedly apply rewrites underneath the W until there are exactly three layers. Assume that fusion is applied as much as possible between each stage and (\ref{eq:kill_quad}) is applied and simplified with (\ref{rule:K0}) to remove $\raisebox{-8pt}{\xsq}$ whenever possible. Then for as long as there are at least 4 layers, we can apply one of the following rewrites:
        \begin{enumerate}[resume, label={(\arabic*)}]
            \item $\input{tikz/poly/lemmas/algcases1.tikz}$
            \item $\input{tikz/poly/lemmas/algcases2.tikz}$
            \item $\input{tikz/poly/lemmas/algcases3.tikz}$
            \item $\input{tikz/poly/lemmas/algcases4.tikz}$
            \item $\input{tikz/poly/lemmas/algcases5.tikz}$
        \end{enumerate}

    \medskip
    
    Clearly, we can only stop applying these rules once $A$ is a sum of products of copies. Steps (2) and (3) ensure the top of $A$ has such a structure and steps (4) - (6) ensure that there is nothing beneath the $\lowerbox{\coWs}$'s . To see that this will always terminate, observe that (2) and (3) preserve the depth of $A$ while (4), (5), (6) all decrease it. (2) and (3) can only be applied a finite number of times before another simplification must be used. So repeatedly applying these rewrites must eventually shrink the depth down to $3$, as desired. Finally, to put $A$ in PNF we must:
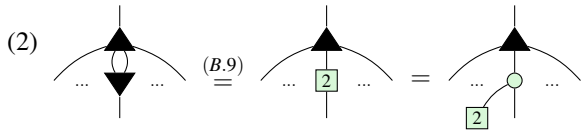
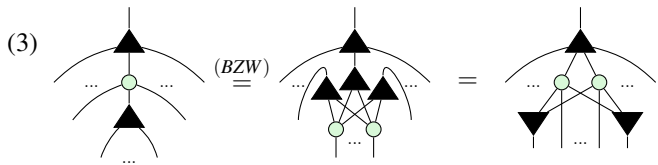
    \begin{enumerate}[resume, label={(\arabic*)}]
        \item Collect terms: whenever there are two boxes connected to exactly the same set of $\raisebox{-5pt}{\coWs}$'s, use (\ref{eq:cpk_add}) to fuse them together. 
        \item Pad: use (\ref{eq:zerobox}) to insert $\raisebox{-5pt}{\numbergate[0]}$ for any connectivities that do not exist in $A$.
        \item Reorder: use (\ref{rule:Sym}) to reorder coefficients into the canonical order.
    \end{enumerate}

    Step (7) ensures that every $\raisebox{-5pt}{\zspids}$ has unique connectivity. Step (8) ensures there are exactly $2^n$ coefficients so that step (9) can order them in the appropriate way. 

    Thus $A$ has been written in PNF, completing the proof.
    
\end{proof}

\textbf{Proof of Theorem \ref{thm:iso}}
\begin{proof}
    
    First, we show $\phi_n$ is a homomorphism, i.e. \begin{equation}
        \forall p, q \in \polyring, \phi_n(p + q) = \phi_n(p) \boxplus \phi_n(q) ,\quad \phi_n(p \times q) = \phi_n(p) \boxtimes \phi_n(q)
    \end{equation}
    The strategy for the proof will be an induction on $n$. Here, the yellow shaded regions serve no purpose other than to make the subdiagram that was rewritten in the preceding or following step easier to visually identify.

    \medskip
    
    \textbf{Base case:}
    We have not defined controlled states for $n=0$, so the base case begins with $n=1$.
    Let $p, q \in \polyring[1]$. Write as $p(x_1) = a_0 + a_1x_1, q(x_1) = b_0 + b_1x_1$, where $a_0, a_1, b_0, b_1 \in \mathbb{C}$. Then since $p + q = a_0 + b_0 + (a_1 + b_1)x_1$,
    \begin{equation}
        \input{tikz/poly/homproof/hombaseadd.tikz}
    \end{equation}

    Meanwhile, since $p \times q = a_0b_0 + (a_0b_1 + a_1b_0)x_1$, 
    \begin{equation}
        \input{tikz/poly/homproof/hombasetimes.tikz}
    \end{equation}

    Completing the base case.

    \medskip

    \textbf{Inductive step:}

    Let $Hom(n)$ assert than $\phi_n$ is a homomorphism.  Then for the inductive step we wish to prove that $\forall n, Hom(n) \implies Hom(n+1)$.

    The proof relies on the recursive definition of $R[x_1, x_2] = R[x_1][x_2]$, for any ring $R$, to rewrite an arbitrary polynomial $p(x_1, ..., x_{n+1}) = a_0 + a_1x_{n+1} + ... + a_{2^{n+1}-1}x_1x_2...x_{n+1} \in \polyring[n+1]$ as $p(x_{n+1}) = p_0 + p_1x_{n+1}$, where $p_0, p_1 \in \polyring$. This allows the $p_i$ to be treated similarly to the scalars in the base case. To emphasise this, they will be drawn in green boxes. To help distinguish when an operation is covered by the inductive hypothesis, the wires for variables $x_1, ..., x_n$ will be drawn in blue, while the $x_{n+1}$ wires will be drawn in black. Thus the inductive hypothesis states that:
    \begin{equation}
        \input{tikz/poly/homproof/ih1.tikz}
        \tag{IH1}\label{eq:ih1}
    \end{equation}
    \begin{equation}
        \input{tikz/poly/homproof/ih2.tikz}
        \tag{IH2}\label{eq:ih2}
    \end{equation}

    Let $p(x_{n+1}) = p_0 + p_1x_{n+1}$ and $q(x_{n+1}) = q_0 + q_1x_{n+1}$, where $p_0, p_1, q_0, q_1 \in \polyring$.
    According to our hypothesis,
    \begin{equation}
        \phi_{n\plu 1}(p) = \phi_{n\plu 1}(p_0) \boxplus \left(\phi_{n\plu 1}(p_1) \boxtimes \phi_{n\plu 1}(x_{n\plu 1}) \right)
    \end{equation}
    and likewise for $\phi_{n\plu 1}(q)$. We first verify correctness of the constructed controlled diagram
    \begin{equation}\label{eq:ctrlp}
        \input{tikz/poly/homproof/ctrlp.tikz}
    \end{equation}
    because it results in $\ket{0...0}$ when the control is $\ket{0}$, and in the state corresponding to $p$ when the control is $\ket{1}$.
    \begin{equation}\label{eq:ctrlp0}
        \input{tikz/poly/homproof/ctrlp0.tikz}
    \end{equation}
    \begin{equation}\label{eq:ctrlp1}
        \input{tikz/poly/homproof/ctrlp1.tikz}
    \end{equation}
    
    Applying this in the inductive step for constructing a controlled diagram of $p + q$ from those of $p$ and $q$:
    \begin{equation}
        \input{tikz/poly/homproof/sa.tikz}
    \end{equation}

    Similarly, for multiplication:
    \begin{equation}
        \input{tikz/poly/homproof/st.tikz}
    \end{equation}

    \medskip

    This completes the inductive step, proving that $\forall n > 1$, $\phi_n$ is a homomorphism.

    \bigskip

    Finally, to see $\phi_n$ is an isomorphism, we use Theorem \ref{thm:uni_pnf} to write an arbitrary controlled state in PNF:
    \begin{gather}
        \begin{bmatrix}
            1 & a_0 \\ 0 & a_1 \\ \vdots & \vdots \\ 0 & a_{2^{n}-1}
        \end{bmatrix}
        \quad = \quad \input{tikz/poly/pnf.tikz}
    \end{gather}

    Then all we have to do is interpret it as the image of a polynomial:
    \begin{gather}
        \input{tikz/poly/pnf.tikz} \quad ~=~ \quad \input{tikz/poly/homproof/iso4.tikz} \neweqline ~=~ \phi_{n}(a_0) + \phi_{n}(a_1x_{n}) + ... + \phi_{n}(a_{2^{n}-1} x_1x_2...x_{n}) \neweqline ~=~ \phi_{n}(a_0 + a_1x_{n} + ... + a_{2^{n}-1}x_1x_2...x_{n}) 
    \end{gather}
\end{proof}

\end{document}